\documentclass[10pt]{article}
\usepackage[total={8.5in,11in}, top=0.75in, left=0.75in, right=0.75in, bottom=0.75in, text={7in,9.5in}]{geometry}
\usepackage{url}

\usepackage{algorithmic}
\usepackage[dvips]{graphicx}
\usepackage{epsfig}
\usepackage{amsmath, amsthm}
\usepackage{color}

\newtheorem{theorem}{\indent Theorem}[section]
\newtheorem{lemma}{\indent Lemma}[section]

\newtheorem{proposition}{\indent Proposition}[section]
\newtheorem{corollary}[theorem]{\indent Corollary}

\def\red#1{{#1}}
\def\blue#1{{#1}}
\def\green#1{{#1}}
\def\blue#1{{#1}}
\xdef\old#1{}

\providecommand{\com[2]}{\begin{tt}[#1: #2]\end{tt}}

\begin{document}

\title{Distributed anonymous function computation in information fusion
and multiagent systems\thanks{The authors are with the Laboratory for Information
and Decision Systems, Massachusetts Institute of Technology, Cambridge, MA, {\tt\small jm\_hend@mit.edu, alex\_o@mit.edu, jnt@mit.edu.}
This research
was supported by the National Science Foundation under grant
ECCS-0701623, and
postdoctoral fellowships from the F.R.S.-FNRS (Belgian Fund for
Scientific Research) and the B.A.E.F. (Belgian American Education
Foundation).
}}

\author{Julien M. Hendrickx, Alex Olshevsky, John N. Tsitsiklis}
\maketitle


\begin{abstract}
\noindent
We propose a model for \red{deterministic} distributed function computation by a
network of identical \green{and} anonymous nodes, \green{with} bounded computation
and storage \green{capabilities} that do not scale with the network size.
Our goal is to \green{characterize the class of functions
that can be computed within this model.}  In our main result, we exhibit a class of
non-computable functions, and prove that every function outside
this class can at least be approximated.
\blue{The problem of computing averages in a distributed manner plays a central role in our development.}
\end{abstract}

\section{Introduction}\label{s:intro}
\thispagestyle{empty}

The goal of many multi-agent systems,
\green{distributed} computation \green{algorithms} and
decentralized data fusion methods is to have a
set of \green{nodes compute} \blue{a} common value based on
\green{initial values or observations at each node. Towards this purpose,
the nodes,} which we will \green{sometimes refer to as agents,} perform some internal
computations and repeatedly communicate with each other. Let us
consider some examples.

\vspace{5pt}
\noindent {\bf (a) Quantized consensus:} Suppose \green{that} each agent begins
with an integer value $x_i(0) \green{\in \{0,\ldots,K\}}$. We would like \green{the agents to end up, at some later time, with values $y_i$ that are almost equal, i.e.,
$|y_i - y_j| \leq 1$, for all $i,j$, while preserving the sum of the values, i.e., $\sum_{i=1}^n x_i(0) = \sum_{\blue{i=1}}^n
y_i$.}
This is the so-called quantized \red{averaging} problem which has
received considerable attention recently; see
\cite{KBS07,FCFZ08,ACR07, KM08}. It may be viewed as the problem
of computing the function $(1/n) \sum_{i=1}^n x_i$, rounded
to the nearest integer.

\noindent {\bf (b) Distributed hypothesis testing:} Consider $n$
sensors \green{interested in deciding} between two hypotheses,
$H_0$ and $H_1$. Each sensor collects measurements and makes a
preliminary decision \red{$x_i\in\{0,1\}$} in favor of one of the hypotheses. The
sensors would like to
\green{make a final decision by majority vote, in which case they need to compute the indicator function of the event
$\sum_{i=1}^n x_i \geq n/2$,
in a distributed way.
Alternatively, in a weighted majority vote, they may be interested in
computing the indicator function of the event
$\sum_{i=1}^n x_i \geq 3n/4$.}




\noindent {\bf (c) Solitude verification:} \green{This is the problem of}
verifying that at most \green{one} node in the network has a given state. \green{This problem is of interest if we want to avoid
simultaneous transmissions over a common channel
\cite{GFL83}, or if we want to maintain a single leader
(as in motion
coordination  --- see for example \cite{JLM03}) Given \blue{$K$}
possible states, \green{so that} $x_i \in \{1,\ldots,\blue{K}\}$, \green{solitude verification} is
equivalent to \green{the problem of} computing the binary function which is equal to 1 if and only if $|\{i: x_i=1\}| =1$.}
\vspace{5pt}

\red{There are numerous methods that have been proposed for solving problems such as the above. (See for example the vast and growing literature on consensus and averaging methods.) Oftentimes, different algorithms involve different computational capabilities on the part of the agents, which makes it hard to talk about ``the best'' algorithm. At the same time, simple algorithms (such as setting up a spanning tree and aggregate information by progressive summations over the tree) are often ``disqualified'' \blue{because they require} too much coordination or global information. One then realizes that a sound discussion of such issues requires the specification of a precise model of computation, followed by a systematic analysis of fundamental limitations under  any given model. This is precisely the objective of this paper: \blue{to} propose a particular model, and to characterize the class of computable functions under this model.}

%
%


\red{Our model provides an abstraction for the most common requirements for distributed algorithms in the sensor network literature. It is somewhat special because (i) it does not allow for randomization; (ii) it does not address the case of time-varying interconnection graphs; such extensions are left for future research. Qualitatively speaking, our model includes the following features.}

\noindent \textbf{\red{Identical} agents:} Any two agents with the same number of neighbors must run the same algorithm.

\smallskip

\noindent\textbf{Anonymity:}
\green{An agent can distinguish its neighbors using its own, private, local identifiers. However, agents do not have global identifiers.}

\smallskip

\noindent\textbf{Absence of global information:} Agents have no
\green{global information, and do not even have an upper bound on the total number of nodes. Accordingly, the algorithm that each agent is running is independent of the
network size and topology.}\smallskip

\noindent\textbf{Convergence:} Agents hold
an estimated output, and this \green{estimate} must converge to a
desired value which is \green{generally} a function of all agents'
\green{initial} observations or \green{values}. \green{In particular, for the case of discrete outputs,}
all agents must eventually \green{settle on} the desired value.
\green{On the other hand, the agents do not need to be aware
of such termination, which is anyway impossible in}
the absence of \green{any} global information \cite{ASW88}.

\subsection{Goal and Contribution}

We \green{provide} in this paper a general model of decentralized
anonymous computation \green{with the above described features,} and
characterize the type of functions of the initial \green{values} that
can be computed. To keep our model simple, we only consider
deterministic \green{and} synchronized agents exchanging messages on a fixed \green{bi-directional}
network, with no time-delays or \green{unreliable transmissions.} Agents are modelled
as finite automata, \green{so that their individual capabilities remain bounded as the number of agents increases.}

We prove that if a function is
computable under \green{our model,} then its value
only depends on the \green{frequencies of the different possible initial values.}
For example, if the initial
\green{values} $x_i$ only take values $0$ and $1$, a computable
function necessarily only depends on $p_0:= |
\{i:x_i=0\}|/n$ and $p_1:=|\{i:x_i=1\}|/n$. In particular,
determining the number of nodes, or \green{whether} at least two nodes have an
initial \green{value} \blue{of} $1$ is \green{impossible.}

Conversely, we prove that if a function only
depends on the \green{frequencies of the different possible initial values (and is measurable)}, then the
function can at least be approximated \green{with any given}  precision, except possibly on a set of frequency vectors of  arbitrarily small volume.
Moreover, if the dependence on \green{these frequencies} can be expressed by
a combination of linear inequalities with rational coefficients,
then the function is computable exactly. In particular,
the functions involved in the quantized consensus and distributed hypothesis testing examples are
computable,
whereas the function involved in solitude verification is not.
Similarly,
statistical measures such as the standard deviation and the
kurtosis can be approximated with arbitrary precision.

Finally,
we show that with infinite memory, \red{the frequencies of the different values} (i.e., $p_0,p_1$ in the binary case) are computable.

\subsection{Overview of previous work}

There \green{is a large literature} on distributed function computation in
related models of computation. A common model in the distributed
computing literature involves the requirement that \green{all} processes
terminate when the \green{desired output is produced.} A consequence of the
termination requirement is that nodes typically need to know the
network size $n$ (or an upper bound on $n$) to compute any non-constant functions. We refer
the reader to \cite{A80, ASW88, YK88, KKB90, MW93} for \green{some}
fundamental results in this setting, \green{and to \cite{FR03} for} an excellent summary of
known results.

Similarly, the biologically-inspired ``population algorithm" model of distributed computation
has some features in common with our model, namely finite-size agents and
lack of a termination condition;
\green{see \cite{AR07} for a very nice summary of known results.}
However, \green{these models involve a somewhat different type of agent interactions from the ones we consider.}

We \green{note} that the impossibility of computing $p_1$ without any memory was shown
in \cite{LB95}. Some experimental memoryless algorithms
were proposed \green{in the physics literature} \cite{GKL78}.
Randomized algorithms for computing particular functions were
investigated in \cite{KBS07, BTV08}. We also point the reader \green{to} \blue{the}
literature on ``quantized averaging,'' which often tends
to involve similar themes \cite{FCFZ08, KM08, ACR07, CB08}.

Several papers quantified the performance of simple
heuristics for computing specific functions, typically in
randomized settings. We refer the reader \blue{to} \cite{HP01} and
\cite{PVV08}, which studied simple heuristics \green{for computing} the
majority function. A deterministic algorithm for computing the
majority function (and some more generalized functions) was
proposed in \cite{LBRA04}.

Semi-centralized versions of the problem, in which the nodes
ultimately transmit to a fusion center, have often be considered
in the literature, e.g., for distributed statistical
inference  \cite{MK08} \green{or detection}
\cite{KLM08}.
The papers \cite{GK05},
 \cite{KKK05}, and \cite{YSG07} consider the complexity of computing a function and
communicating \green{its value} to a sink node. We refer the
reader to the references therein for an overview of \green{existing}
results in such semi-centralized settings.

Our results differ from previous works in several key respects:
(i) Our model, which involves totally decentralized computation,
deterministic algorithms, and constraints on memory and computation
resources at the nodes, but does not require the nodes to know
when the computation is over, \green{does not seem to have been studied before.}
(ii) Our focus is on identifying computable \green{and} non-computable
functions \blue{under our model,} and we achieve a nearly tight
separation.


\section{Formal description of the model}

The system consists of (i) a communication graph $G=(V,E)$, which
is bidirectional (\green{i.e.,} if $(i,j) \in E$, then $(j,i) \in E$);
(ii) a \emph{port labeling} whereby edges outgoing  from node $i$
are labeled by \emph{port numbers} in the set $\{1,2,\ldots,{\rm degree}(i)\}$; (iii)
a family of finite automata $(A_d)_{d = 1,2,3,\ldots}$.
\green{(The automaton $A_d$ is meant to describe the behavior of a node with degree $d$.)}

The state of the automaton $A_d$ is \green{a} tuple $(x,z,y, m_1,
\ldots, m_d)$;  we will call $x \in X\blue{=\{0,1,\ldots,K\}}$ the initial \green{value}, $z\in
\red{Z_d}$ the internal memory state,
 $y\in Y$ the output or estimated answer, and $m_1, \ldots, m_d \in M$ the
 messages. The sets $X,Y,Z_d, M$ are \red{assumed
finite, unless there is a statement to the contrary.
 Furthermore, we assume that the number of bits that can be stored at a node is proportional to the node's degree; that is, $\log |Z_d| \leq Cd$, for some absolute constant $C$. (Clearly, at least this much memory would be needed to be able to store the messages received at the previous time step.)}  We will also assume that $\emptyset$ is an element of
\blue{the above defined sets $Y$, $Z_d$, and $M$.} The transition law $A_d$ \green{maps} $X \times \blue{Z_d} \times Y \times M^d$
to $X \times \blue{Z_d} \times Y \times
M^d$: $\left[x,z,y;(m_1,\dots,m_d)\right]$ is \green{mapped} to $\left[x, z',
y';(m'_1,\dots, m'_d)\right].
$ In words, the automaton creates a new memory state, output, and \green{(outgoing)} messages at
each iteration, but does not change the initial \blue{value.}

We construct a dynamical system out of the above elements as
follows. \green{Let $d(i)$ be the degree of node $i$.} Node $i$ \green{begins with an initial value
$x_i\in X$;} it implements the automaton $A_{d(i)}$, \green{initialized} with $x=x_i$, and with
$z=y=m_1=\ldots=m_d = \emptyset$.
\green{We use $S_i(t)=
[ \blue{x_i},y_i(t),z_i(t),m_{i,1}(t), \ldots, m_{i, d(i)}(t)]$
to denote the}
state
of the automaton implemented by agent $i$ at round $t$.
Let $j_1, \ldots,
j_{d(i)}$ be an enumeration of the neighbors of $i$, and let $p_k$
be the port number of the link $(j_k,i)$.   The evolution of the
system is then described by  
$$\left[x_i,
z_i(t+1),y_i(t+1); m_{i,1}(t+1), \ldots, m_{i, d(i)}(\red{t+1}) \right]  =
A_{d(i) }\left[x, z_i(t), y_i(t) ; m_{\blue{j_1, p_1}}(t),
\ldots, m_{j_d, p_d}(t) \right].
$$  
In words, the
messages ``sent'' by the neighbors of $i$ into \green{ports
leading} to $i$ are used to transition to a new state and
create new messages that  $i$ ``sends'' to its neighbors at the next
round. We say that $y^*\in Y$ is the \emph{\green{final output}} of
this dynamical system if there is a time $t'$ such that $y_i(t) =
y^*$ for every $i$ and $t\geq t'$.

Consider now a family of functions $(f_n)_{n\in
\mathcal{N}}:X^n\to Y$. We say that such a
family is \emph{computable} if there exists a family of automata
$(A_d)_{d=1,2,\ldots}$ such that \blue{for any $n$,} for any connected graph
$G=(V,E)$ \blue{with $n$ nodes,} any port labelling, and any set of initial conditions
$x_1,\dots,x_{\blue{n}}$, the \green{final output} of the above system is
always $f_{\blue{n}}(x_1,\dots,x_n)$.

In some results, we will also refer to function families $(f_n)_{n
\in \mathcal{N}}$ {\em computable with infinite memory},  by which
we mean that the internal memory \blue{sets $Z_d$} and output set $Y$ are
countable, the rest of the model being unaffected.

We study in the sequel the general function computation problem:
What \green{families} of functions are computable, and
how can we design the automata $A_d$ to compute them?

\section{Necessary condition for computability}\label{sec:nec_conditions}

Let us first state the following lemma which can easily be proved
by induction on time.

\begin{lemma} \label{switching}
Suppose that $G=(\{1,\ldots,n\},E)$ and $G'=(\{1,\ldots,n\},E')$
are isomorphic, that is, there exists a permutation $\pi$ such
that $(i,j) \in E$ if and only if $(\pi(i),\pi(j)) \in E'$.
Further, suppose that the \green{port label at node $i$ for the edge leading to} $j$ in $G$ is the
same as the \green{port label at node $\pi(i)$ for the edge leading to} $\pi(j)$ in $G'$.  Then, \green{the state
$S_i(t)$ resulting} from \blue{the initial values $x_1,\ldots,x_n$} with
the graph $G$ is the same as \blue{the state} $S_{\pi(i)}(t)$
\blue{resulting from the initial values $x_{\pi(1)},\ldots,x_{\pi(n)}$}
 with the graph $G'$.
\end{lemma}

\begin{proposition} \label{symmetry}
Suppose that the family $ \{f_1(x_1), f_2(x_1,x_2), f_3(x_1,x_2,x_3),
\ldots \}$ is computable with infinite memory. Then, each $f_i$
is invariant under permutations of its arguments.
\end{proposition}
\begin{proof}
\blue{Let ${\pi_{ij}}$ be permutation that swaps $i$ with $j$; with a slight abuse of notation, we also denote by $\pi_{ij}$} the mapping from \blue{$X^n$ to $X^n$} that swaps the
$i$th and $j$th elements of a vector. We show that for all $x \in
\blue{X^n}$, $f_{\blue{n}}(x) = f_{\blue{n}}( \pi_{ij}( x))$.

We run our \green{distributed} algorithm \green{on} the $\blue{n}$-node
complete graph. Consider two different initial
configurations: (i) starting with the vector $x$; (ii) starting with
the vector $\pi_{ij}(x)$. Let the way each node enumerates his
neighbors in case (i) be arbitrary; \green{in case (ii), let the enumeration be}
\blue{such that the conditions in Lemma \ref{switching} are satisfied, which is easily accomplished.}
Since the limiting value of $y_i$ in one case is $f(x)$ and in
the other is $f(\pi_{ij}(x))$, we obtain $f(x)=f(\pi_{ij}(x))$.
\blue{Since the permutations $\pi_{ij}$ generate the group of permutations, permutation invariance follows.}
\end{proof}

Let \blue{$x\in X^n$.}  We will \green{denote by} $x^2$ the
concatenation of $x$ with itself, and, generally, $x^k$ the
concatenation of $k$ copies of $x$. We now prove that
self-concatenation does not affect the value of a computable
family of functions.

\begin{proposition} \label{repetition} Suppose that the family
$\{f_1(x_1),
f_2(x_1,x_2), f_3(x_1,x_2,x_3), \ldots \}$ is computable with
infinite memory. Then, \red{for every $m\geq 2$}, every sequence
\blue{$x\in X^m$}, and \green{every positive} integer~$k$, \[ f_m(x) = f_{km} (x^k).\]
\end{proposition}

\begin{proof}
Consider a ring of agents \blue{of} size $m$, where the $i$th agent
counterclockwise begins with the $i$th element of $x$; and
consider a ring of size $km$ where the agents
\red{$i,i+m,i+2m,\ldots$}
(counterclockwise) begin with the $i$th element of $x$.  Suppose
further that each node enumerates its two neighbors so that the
neighbor on the left is labelled 1, while the neighbor on the
right is labelled $2$. See Figure \ref{rings} for an example with
$m=3,k=2$ and $x_i=i$.

\begin{center}
\begin{figure}[h] \hspace{5cm}
\begin{minipage}[t]{.45\textwidth}
  \begin{center}        \epsfig{file=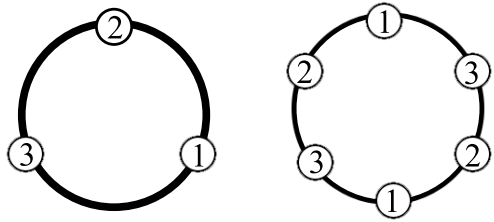, scale=0.35}
  \end{center}
\end{minipage}
\caption{Example of two situations that are \blue{indistinguishable} by
the nodes.} \label{rings}
\end{figure}   \end{center}

Initially, the state of node $i$ in the first ring is exactly the
same as the state of the nodes \blue{$j=i,i+m,i+2m,\ldots$} in the second ring. We show by induction \green{that this property} must hold at all
times $t$. \blue{(To keep notation simple, we assume, without any real loss of generality, that $i\neq 1$ and $i\neq m$.)}

Indeed, suppose this \green{property} holds up to time $t$. At time $t$, node $i$ in
the first ring receives a message from node $i-1$ and a message
from node $i+1$; and in the second ring, node \red{$j$} satisfying $\red{j \ (\mbox{mod } m)  } =
i$ receives one message from $j-1$ and $j+1$.
Since $\red{j-1 \ (\mbox{mod } m)  } =
i-1$ and $\red{j+1 \ (\mbox{mod } m)  } =
i+1$,
the states of \green{$j-1$ and $i-1$} are identical at time $t$,
\green{and similarly for
and $j+1$ and $i+1$. Thus the messages received by $i$ (in the first ring) and $j$ (in the second ring) at time $t$ are
identical. Since $i$ and $j$} were in the same state at time $t$, they
must be in the same state at time $t+1$. This proves that they are
in the same state forever.

It follows that $y_i(t) = y_j(t)$ for all $t$, whenever
$\red{j \ (\mbox{mod } m)  } =
i$, and therefore $f_m(x)=f_{km}(x^k)$.
\end{proof}

We can now prove our main \green{negative} result, stating that if a family of
functions is computable, then the value of the function only
depends on the \green{frequencies of the different possible initial values.}
\green{
We define $D=\{(p_1,\dots,p_K)\in [0,1]^K : \sum_{k=1}^Kp_k =1 \}$, which we call
the \emph{proportion set}. We say that a function $h:D\to Y$
corresponds to a family $(f_n: \blue{X}^n\to Y)$
if for every $x\in \blue{X}^n$,
$$
f(x_1,\dots,x_n)=
h\left(p_1(x_1,\ldots,x_n),\dots,p_K(x_1,\ldots,x_n)\right),
$$
where
\[ p_k(x_1,\ldots,x_n) = |\{ \blue{i} ~|~ x_{\blue{i}} = \blue{k}, \}|/n,\] so that
$p_k(x_1,\ldots,x_n)$ is the frequency of occurrence of the initial value $k$.
In this case, we say that the family $(f_n)$ is \emph{proportion-based.}}

\begin{theorem} \label{unbounded-bound}
Suppose that the family $(f_n)$ is computable with infinite memory. Then, this family is proportion-based.
\end{theorem}
\begin{proof}
Let $x$ and $y$ be two sequences of $n$ and $m$
elements, respectively,  such that $p_k(x_1,\dots,x_n)= p_k(y_1,\dots,y_m)=: \hat
p_k$ for $k=1,\dots,K$, that is, the number of occurrences
of $k$ in $x$ and $y$ are $n\hat p_k$ and $m\hat
p_k$, respectively. Observe that \blue{for any $k\in X$, the vectors}  $x^m$ and $y^n$ have the same number $mn$ of
elements, and both contain $mn\hat p_k$ occurrences of \green{$k$.}
The sequences $y^n$ and $x^m$ can thus be obtained
from each other by a permutation, which by Proposition
\ref{symmetry} implies that $f_{nm}(x^m) = f_{nm}(y^n)$. It
then follows from Proposition \ref{repetition} that $f_{nm}(x^m) = f_n(x)$
and $f_{mn}(y^n)=f_m(y)$, so that $f_n(x) = f_m(y)$. This proves
that the value of $f_n(x)$ is determined by the vector
$(p_1(x),\dots,p_n(x))$.
\end{proof}

\def\texitem#1{\par
\noindent\hangindent 25pt
\hbox to 25pt {\hss #1 ~}\ignorespaces}

The following examples illustrate this result.

\texitem{(a)}
The parity function $\sum_{i=1}^n x_i ~ ({\rm mod ~k})$ is not
computable, for any $k>0$.
\texitem{(b)}
In a binary setting \blue{$(X=\{0,1\})$}, checking whether the number of nodes with
\green{$x_i=1$ is at least $10$ plus} the number of nodes with $x_i=0$ is
not computable.
\texitem{(c)}
Solitude verification, i.e.
checking whether $|i: \{x_i=0\}|=1$, is not computable.
\texitem{(d)}
An aggregate difference functions such as
$\sum_{i<j} |x_i - x_j|$ is not computable, even \green{calculated modulo $k$.}


\section{Reduction of generic functions to the computation of averages}\label{sec:reduction_to_avg}

In this section, we show that the computability question for large
classes of functions reduces to the computability question for a
particular averaging-like function. Namely, we will make use of
the following theorem.

\begin{theorem}\label{thm:average_computation} Let $X = \{0,\ldots,K\}$ and
define $Y$ to be following set of single-point sets and intervals:
\[ Y = \{ \{0\}, (0,1), \{1\}, (1,2),  \ldots,
\{K-1\}, (K-1,K), \{K\} \}\]
\green{(or equivalently, an indexing of this finite collection of intervals).}
Let $(f_n)$ be the following family
of functions: $f_n$ maps $(x_1,x_2,\ldots,x_n)$ to the element
of $Y$ which contains the average $\sum_i x_i/n$. Then, the family
$(f_n)$ is computable.
\end{theorem}

The proof of Theorem \ref{thm:average_computation} is fairly involved and too long
to be included in this extended abstract; however, we give an informal description of the algorithm for
computing $f_n$ in Section \ref{sec:comput_average}.
In this section, we show that Theorem
\ref{thm:average_computation} implies \green{the computability of} a large class of functions.
\green{We will say that
a function $h$ on the proportion set
is \emph{computable} if it it corresponds to a proportion-based computable family $(f_n)$.  The
level sets of $h$ are defined as the sets $L(y) = \{ \red{p} \in D ~|~ h(p) = y\}$,
for $y \in Y$.}

\begin{theorem}[Sufficient condition for computability] \label{bounded-alg}
Let $h$ be a function from the proportion set $D$ to $Y$. \green{Suppose that} every level set $L(y)$ can be written as a
finite union,
\[ L(y) =  \bigcup_k C_{i,k},\] where each $C_{i,k}$ can in turn be
written as a finite intersection of \green{linear} inequalities \green{of the form}
\[ \alpha_1 p_1 + \alpha_2 p_2 + \cdots + \alpha_K p_K \leq \alpha,
\] or
\[ \alpha_1 p_1 + \alpha_2 p_2 + \cdots + \alpha_K p_K < \alpha,
\]
with rational \green{coefficients}  $\alpha,\alpha_1,\ldots,\alpha_K$. Then, $h$
is computable.
\end{theorem}

\begin{proof}
Consider one such linear inequality. Let $P$ be the set of indices
\green{$i$ for which} $\alpha_i \geq 0$. Since all coefficients are rational, we
can clear the denominators and rewrite \green{the inequality} as
\begin{equation}\label{eq:integ_equiv_density_ineq}
\sum_{k \in P} \beta_k p_k - \sum_{i \in P^c} \beta_k p_k \leq
\beta , \end{equation} for \red{nonnegative} integers $\beta_k$ and
$\beta$. Let $\chi_k$ be the indicator function \green{associated with initial
value} $k$, i.e., $\chi_k(i)=1$ if $x_i = k$, and $\chi_k(i) =0$
otherwise, so that $p_k=\frac{1}{n}\sum_{i}\chi_k(i)$.
\green{Then, (\ref{eq:integ_equiv_density_ineq}), becomes} $$\frac{1}{n} \sum_{i=1}^n \left(\sum_{k\in P}\beta_k
\chi_k(i) + \sum_{k\in P^c}\beta_k  (1-\chi_k(i))\right) \leq
\beta + \sum_{k\in \red{P^c}}\beta_k, $$ or \[ \frac{1}{n}\sum_{i} q_i
\leq  q^*, \] where $q_i = \sum_{k\in P}\beta_k \chi_k(i) +
\sum_{k\in P^c}\beta_k (1-\chi_k(i))$ and $q^* =\beta + \sum_{k\in
P^c}\beta_k$.

To determine if the latter inequality is satisfied, each node can
compute $q_i$ and $q^*$, \green{and then apply a distributed algorithm that
computes $\frac{1}{n}\sum_{i=1}^n q_i$, which is possible by virtue of} Theorem
\ref{thm:average_computation}. To check any finite collection of
inequalities, the nodes can perform the computations for each
inequality in parallel.

To compute $h$, the nodes simply need to check which \green{set} $L(y)$ the
\green{frequencies} $p_1,\ldots,p_K$ lie in, and this can be done by
checking the inequalities defining each $L(y)$.
\green{All of these computations} can be accomplished
with finite automata: indeed, we do nothing more than run finitely
many copies of the automata provided by Theorem
\ref{thm:average_computation}, one for each inequality.
\end{proof}

Theorem \ref{bounded-alg} shows the computability of functions $h$
whose level-sets \green{can} be defined by linear inequalities with
rational coefficients.
\red{On the other hand, it is clear that not every function $h$ can be computable.} \blue{(This can be shown by a counting argument: there are uncountably many possible functions $h$, but for the special case of bounded degree graphs, only countably possible algorithms.)} \green{Still, the next lemma shows that the set of computable functions is rich enough, in the sense that such}
functions
can approximate any measurable function.

We will call a set of the form $\prod_{\red{k}=1}^K (a_k,b_k)$, \green{with every $a_k,b_k$ rational,} {\em a \red{rational open}
box}, where $\prod$ \green{stands for Cartesian}
product. A function that can be written as a finite sum $\sum_i
a_i \red{1_{B_i}}$, where \green{the $B_i$ are rational open boxes and the $1_{B_i}$ are the associated indicator functions,} will be referred to as a \emph{box
function.} \green{Note that box functions are computable by Theorem
\ref{bounded-alg}.}

\begin{corollary}\label{cor:approx_set_epsilon}
If every level set of a function $h:D\to Y$ on the proportion set
is Lebesgue measurable, then, for every $\epsilon >0$, there exists a
computable box function $h_\epsilon:D\to Y$ such that the
set $\{ h \neq h_{\epsilon} \}$ has measure at most $\epsilon$.
\end{corollary}
\begin{proof}
The proof relies on the following elementary result from measure
theory.
Given a Lebesgue measurable set $E \in [0,1]^\red{K}$ and some $\epsilon
> 0$, there \green{exists a set $E'$ which is a finite union of
disjoint open boxes, and which satisfies}
\[ \mu( (E - E') \cup (E'-E) ) < \epsilon,\] where $\mu$ is
the Lebesgue measure.
\red{By a routine argument, these boxes can be taken to be rational.
By applying this fact to the level sets of the function $h$ (assumed measurable), the function $h$ can be approximated by a box function}  \blue{$h_{\epsilon}$. Since box functions are computable, the result follows.}
\end{proof}

The following corollary \green{states} that quantizations of continuous
functions are approximable.

\begin{corollary}\label{cor:approx_continuous}
If a function $h:D\to [L,U]\subseteq \Re$ is continuous, then for
every $\epsilon>0$ there exists a computable function
$h_\epsilon:D \blue{\to  [L,U]}$ such that
$\|h-h_\epsilon\|_\infty < \epsilon$
\end{corollary}
\begin{proof}
Since $D$ is compact, $f$ is uniformly continuous. One can
therefore \green{partition} $D$ into a finite number of subsets
$A_1,A_2,\dots,A_q$, that can be \green{described} by linear inequalities
with rational coefficients, so that $\max_{p\in
A_j}h(p)-\min_{p\in A_j}h(p) <\epsilon$ holds for all $A_j$. The
function $h_\epsilon$ is then built by assigning to each $A_j$ an
\green{appropriate} value \green{in} $\{L,L+\epsilon, L+2\epsilon, \dots, U\}$.
\end{proof}

Finally, we show that with infinite memory, it is possible to
recover the exact \green{frequencies $p_k$.}
\red{(Note that this is impossible with finite memory, because $n$ is unbounded, and the number of bits needed to represent $p_k$ is also unbounded.)}
\begin{theorem} \label{thm:infinite_mem_positive}
The vector $(p_1,\ldots,p_K)$ is computable with infinite memory.
\end{theorem}
\begin{proof}
We show that $p_1$ is computable exactly, which is sufficient to
prove the theorem. Consider the following algorithm, parametrized
by a \green{positive integer} \red{$m$.} The initial set $X_m$ will be $\{0,1,\ldots,
m\}$ and the output set $Y_m$ will be as in Theorem
\ref{thm:average_computation}:
$Y_m = \{ \{0\}, (0,1), \{1\}, (1,2), \{2\}, (2, 3),
\ldots, \{m-1\}, (m-1,m), \{m\} \}$. If $x_i=1$, then node sets
its initial value $x_{i,m}$ to $m$; else, the node sets its
initial value $x_{i,m}$ to $0$. \green{The algorithm computes} the
function family $(f_n)$ which maps $X_m^n$ to the element of $Y_m$
containing $(1/n) \sum_{i=1}^n x_{i,m}$, \green{which is possible, by} Theorem
\ref{thm:average_computation}. We
will call this algorithm $Q_m$. Let $y_m$ be its \green{final output.}

The nodes run \green{the algorithms} $Q_m$ for every \green{positive integer value of $m$, in an interleaved manner. Namely, at each time step, a node runs one step of a particular algorithm $Q_m$, according to}
the following order:
\[ Q_1, \hspace{0.2cm} Q_1, Q_2, \hspace{0.2cm} Q_1, Q_2, Q_3, \hspace{0.2cm} Q_1, Q_2, Q_3,
Q_4, \hspace{0.2cm} Q_1, Q_2, \ldots \]

At each time $t$, let \red{$m_i(t)$} be the smallest $m$ \green{(if it exists)} such that \green{the output
\red{$y_{i,m}(t)$} of $Q_m$ at node $i$}  is a \green{singleton} (not an interval).
\green{We identify this singleton with the numerical value of its single element, and we}
set
$y_i(t) = y_{i, m_i(t)}(t)/m_i(t)$. If \green{$m_i(t)$ is undefined},  then $y_i(t)$ is set to some default value.

It follows from the definition of $Q_m$ and from Theorem
\ref{thm:average_computation} that there exists a time after which
the outputs $\red{y_{i,m}}$ of the algorithms $Q_1,\dots,Q_n$ \green{do not change, and are the}
same for every node, \green{denoted $y_m$.} Moreover, at least one \green{of these algorithms} has an integer
output $y_m$. Indeed observe that $Q_n$ computes $(1/n)
\sum_{i=1}^n n1_{x_i=1} = \sum_{i=1}^n 1_{x_i=1}$, which is clearly
an integer. \red{In particular, $m_i(t)$ is eventually well-defined and bounded above by $n$.} We conclude that there exists a time after which the
output of our \green{overall} algorithm is fixed, shared by all nodes, and
different from the default value.

We now argue that this value is indeed $p_1$. \red{Let $m^*$ be} the smallest $m$ for
which the eventual output of $Q_m$ is a single integer $y_m$. Note that
$y_{m^*}$  is the exact average  of the $x_{i,m^*}$, i.e. $y_{m^*} =
\frac{1}{n} \sum_{i=1}^n m^* 1_{x_i=1} = m^* p_1$. \red{For large $t$, we have $y_i(t) = y_{i, m^*}(t)/m^*
=p_1$.}

Finally, it remains to argue that the algorithm described here can
be implemented with a sequence of automata. All the above
algorithm does is run a copy of all the automata implementing
$Q_1, Q_2, \ldots$ with time-dependent transitions. This can be
accomplished with an automaton whose state space is \green{the countable set} $\mathcal{N}
\times \red{\cup_{m =1}^{\infty}\prod_{i=1}^{m} {\cal Q}_i}$, where ${\cal Q}_i$ is the state \green{space} of $Q_i$,
and the set $\mathcal{N}$ of integers is used to keep track of
time.
\end{proof}

To illustrate the results of this section, let us consider again
some examples.

\texitem{(a)} \green{Majority testing between two options is} equivalent
to checking whether $p_1 \leq 1/2$, with alphabet $\{0,1\}$, and is
therefore computable.
\texitem{(b)}
Majority testing when some nodes can ``abstain'' amounts to checking whether $p_1 -
p_2 \geq 0$, with alphabet $\{0,1,{\rm abstain}\}$. This function family is
computable.
\texitem{(c)}
We can ask for the second
most popular value out of four, for example. In this case, the
sets $A_i$ can be decomposed into constituent sets defined by
inequalities such as $p_2 \leq p_3 \leq p_4 \leq p_1$, each of which
obviously has rational coefficients.
\texitem{(d)}
\green{For any subsets $I,I'$ of
$\{1,\ldots,K\}$, the indicator function of the set
where} $\sum_{i \in I} p_i
\geq \sum_{i \in I'} p_i$  is computable. This is equivalent to checking
whether more nodes have a value in $I$ than do in $I'$.
\texitem{(e)}
The \green{indicator functions of the sets defined by} $p_1^2 \leq 1/2$ and
$p_1 \leq \pi/4$ are measurable, so they are approximable. We are
unable to say whether they are computable.
\texitem{(f)}
The \green{indicator function of the set defined by} $p_1 p_2 \leq 1/8$ is approximable, but we are unable
to say whether it is computable.


\section{A sketch of the proof of Theorem \ref{thm:average_computation}}\label{sec:comput_average}

\green{In this section, we sketch} an algorithm for computing the average of integer
initial values, but omit the proof of correctness. We start with an important subroutine \green{that tracks the maximum (over all nodes) of
 time-varying inputs at each node.}

\subsection{Distributed maximum tracking \label{maxtrack}}

Suppose that each node $i$ has \green{a time-varying input} $u_i(t)$ stored in
memory at time $t$, \green{belonging to a finite} set of numbers ${\cal U}$. \green{We assume that, for each $i$, the sequence
$u_i(t)$} must eventually stop changing, i.e., \green{that there exists}
some $T'$ \green{such that}
\[ u_i(t) = u_i(T'), ~~~~~~~~\mbox{ for all  } t \geq T'.\]
(However, the nodes need not be ever aware that $u_i(t)$ has reached its final value.)
Our goal is to develop a distributed algorithm \green{whose output eventually settles on the value $\max_i u_i(T')$. More precisely, each} node $i$ is to maintain a number $M_i(t)$ which must
satisfy the following constraint: \red{for every connected graph} \green{and any allowed sequences $u_i(t)$,} there exists some $T''$ with
\[ M_i(t) = \max_{i=1,\ldots,n} u_i(t), ~~~~~\mbox{  for all  } t \geq
T''. \]

Moreover, node $i$ must also maintain a pointer $P_i(t)$ to
a neighbor or to itself. We will use the notation
$P^2_i(t)=P_{P_i(t)}(t)$, $P^3_i(t)=P_{P_i^2(t)}(t)$, etc.
\green{We require}
the following \green{additional} property, for all $t$ larger than $T''$: for
each node $i$ there exists a node $\green{j}$ and a power $K$
such that for all $k \geq K$ we have $P^k_i =  j$; moreover,
$M_i(t) = u_j(t)$. In other words, by successively following the
pointers $P_i(t)$, one \green{can} arrive at a node with the maximum
value.

\begin{theorem} An algorithm satisfying the above conditions exists and
can be implemented at each node with a finite automaton \green{whose} state \blue{can be stored using} at most $C(\log
|{\cal U}|+d(i))$ bits, for some absolute constant $C$.
\end{theorem}

We briefly summarize the algorithm guaranteed by this theorem. Each node \green{$i$} initially
sets $M_i(0)=u_i(0), P_i(0)=i$. Nodes exchange their values \green{$u_i(t)$} and forward the largest \green{value}
 they have seen; every node sets its \green{estimated maximum} $M_i\green{(t)}$ equal to \green{that largest} value, and sets its pointer $P_i$ to the node
 \green{that forwarded that value to $i$.}  When some $u_i$ changes,
 the corresponding node sends out a reset message, which is \green{then} forwarded by all \green{other} nodes.
\green{The details of the algorithm and its analysis are somewhat involved because we need to make sure that}  the reset  messages do not cycle forever.

\subsection{The averaging algorithm}
We \green{continue} with an intuitive description of \green{the averaging} algorithm.
Imagine the initial integer values $x_i$ \green{as represented by $x_i$} pebbles.
Our algorithm attempts to exchange pebbles between nodes with
unequal number of pebbles so that the overall distribution \green{becomes} more
even. Eventually, either all nodes will have the same number of
pebbles, or some will have a certain number and others just one
more. \red{We let $u_i(t)$ be the current number of pebbles at node $i$; in particular, $u_i(0)=x_i$. An important property of the algorithm is that the total number of pebbles is conserved.}

To match nodes with unequal number of pebbles we use the
maximum tracking algorithm of Section \ref{maxtrack}. Recall
that the algorithm provides nodes with pointers which attempt to
track the \green{location of the} maximal values. When a node with $\red{u}_i$ pebbles comes to
believe in this way that a node with at least $u_i+2$ pebbles
exists, it sends a request in the direction \green{of the latter node} to obtain one or more
pebbles. This request follows \blue{a} path to \blue{a} node with \blue{a} maximal
number of pebbles, until \blue{the request} either gets denied, or gets accepted by a
node with at least $u_i+2$ pebbles.

More formally, the algorithm uses two types of messages:
\begin{enumerate}
\item[(a)] \green{(Request, $r$): This is a request for a
transfer of value. Here, $r$ is an integer that} represents the
\green{current number of pebbles at the}
emitter.

\item[(b)]
\green{(Accept, $w$): This corresponds to} acceptance of
a request, and subsequent transfer of  $w$ \green{pebbles to the requesting node. A request with a value $w=0$} represents a request denial.
\end{enumerate}
\blue{As} part of the algorithm, the nodes run the maximum tracking
algorithm of Section \ref{maxtrack}, as well as a minimum tracking counterpart. In particular, \green{each node $i$ has}
access to the variables $M_i(t)$ and $P_i(t)$ of the maximum
tracking algorithm (recall that these are, respectively, the
\green{estimated maximum and a} pointer to a neighbor). \green{Furthermore,} each node
maintains three additional variables.

\begin{itemize}
\item[(a)]
``mode"$\in$ \{free,blocked\}. Initially, the mode of every
node is free. Nodes become blocked when they are handling
requests.
\item[(b)]
``$Rin_i(t)$'', ``$Rout_i(t)$" are pointers to a
neighbor of $i$ or to itself. They represent, respectively, the
node from which $i$ has received a request, and \green{the node} to which \green{$i$} has
transmitted \green{a request.} Initially,
$Rin_i(0)=Rout_i(0)=\emptyset.$   %
\end{itemize}
The algorithm is described in Figure
\ref{fig:descr_avg_algo}. 

\begin{figure}
\centering

\includegraphics[scale = .38]{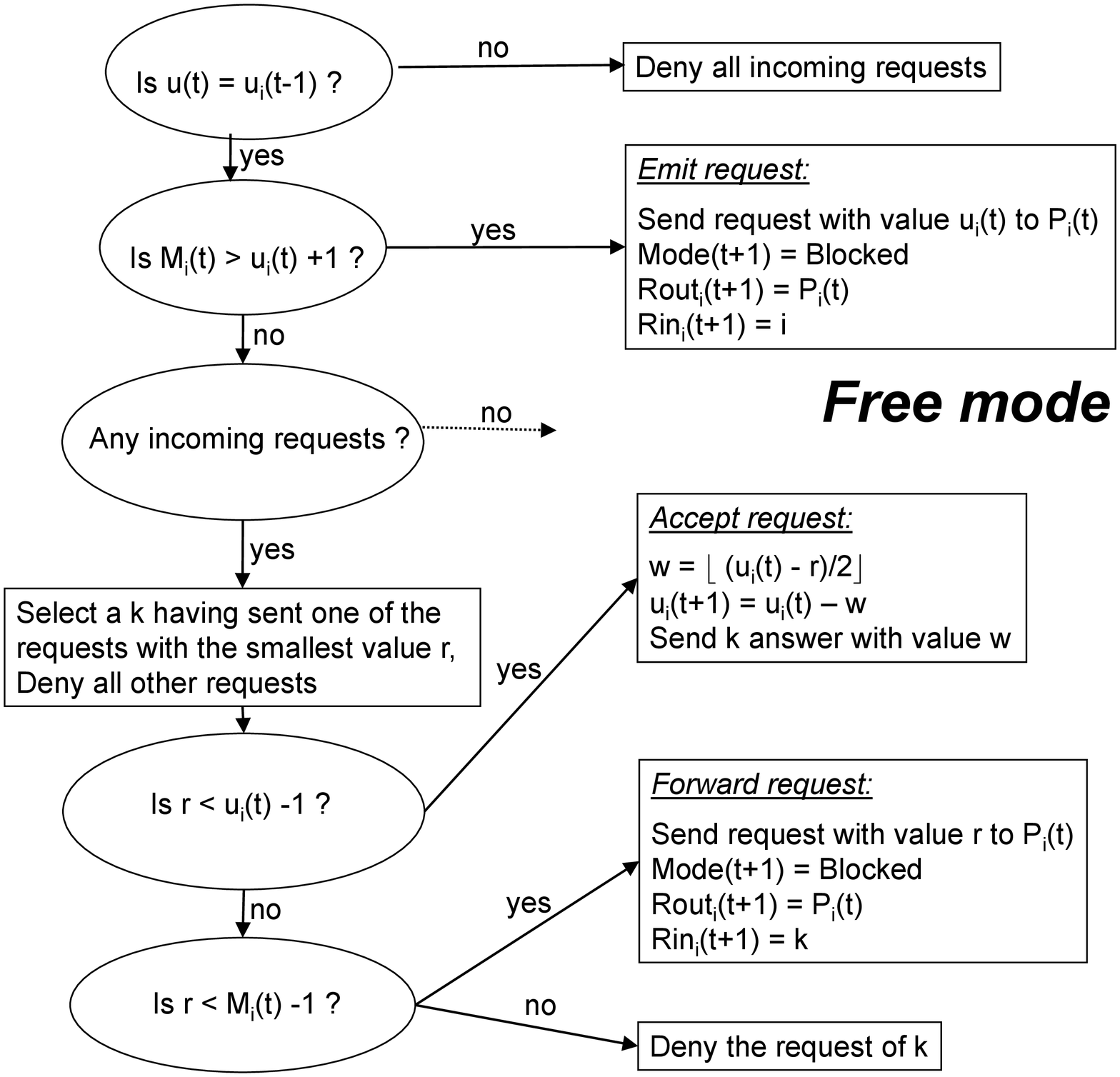}\\
\includegraphics[scale = .38]{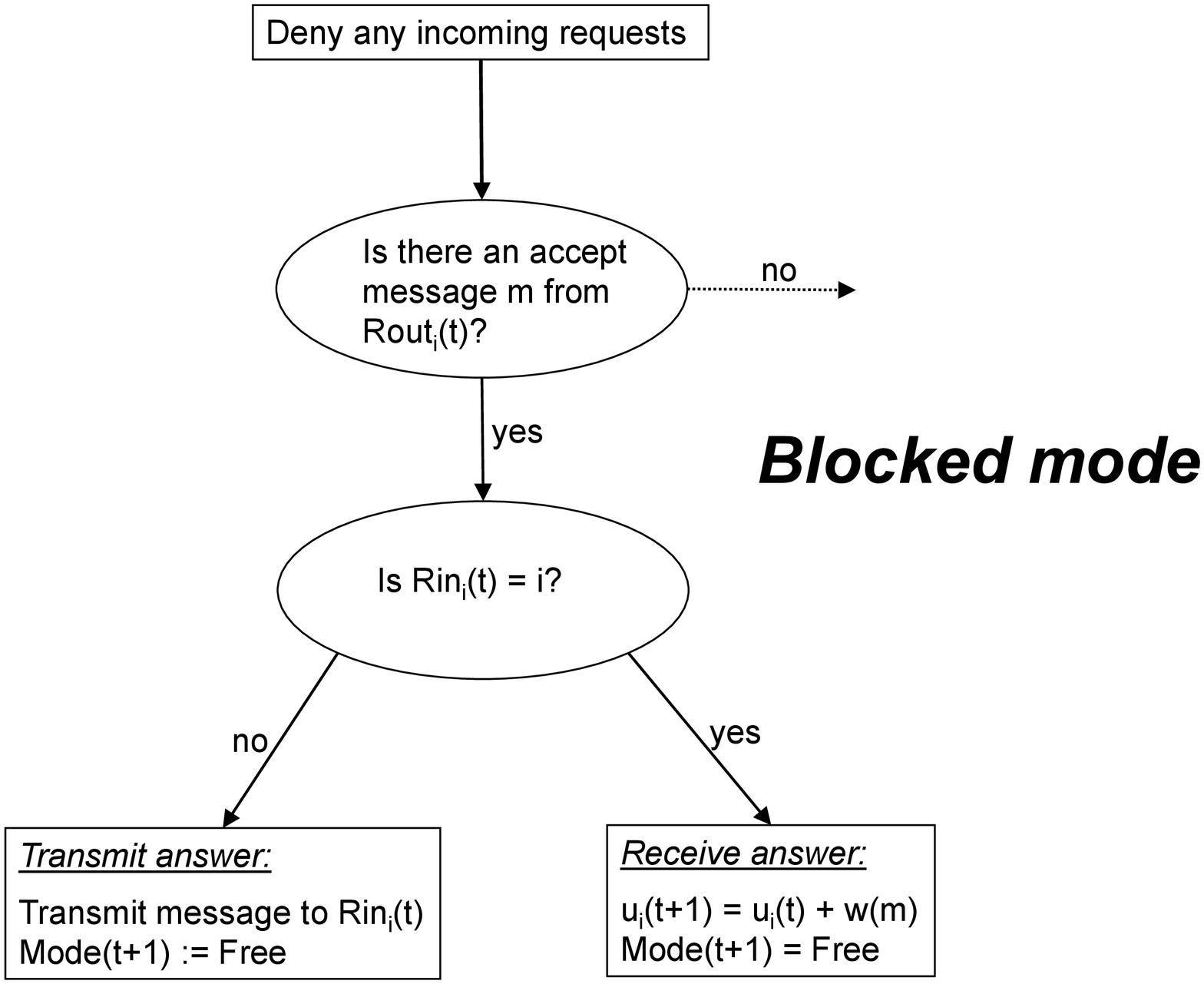}

\caption{Representation of the average computation
algorithm}\label{fig:descr_avg_algo}
\end{figure}

A key step in the proof of Theorem
\ref{thm:average_computation} is the following proposition, whose proof is omitted.

\begin{proposition}\label{pr:last} There is a time $t'$ such that $\red{u}_i(t) =
\red{u}_i(t')$, for all \green{$i$ and} $t \geq t'$. Moreover, \begin{eqnarray*}
\sum_i u_i(t') & = & \sum_i u_i(0)\  \green{=\ \sum_i x_i,} \\
|u_i(\red{t'}) - u_j(\red{t'})| & \leq & 1, \mbox{ for all } i,j.
\end{eqnarray*} \end{proposition}

We now conclude our sketch of the proof of Theorem \ref{thm:average_computation}.
\red{Let $u_i^*$ be the value that $u_i(t)$ settles on.}
It follows from \green{Proposition \ref{pr:last}} that if the average $\bar x$ of
the \green{inputs $x_i$} is integer, then \green{$u_i(t)=u_i^*=\bar x$} will
eventually hold for every $i$. If \green{$\bar x$} is not an integer, then some
node will eventually have $\red{u_i^*}=  \lfloor \bar x \rfloor$
and others $u_i^* =\lceil \bar x \rceil$. Using
the maximum and minimum computation algorithm, nodes will
eventually have a correct estimate of \red{$\max_i u_i^*$ and $\min_i
u_i^*$,} because \blue{each} $u_i(t)$ \green{converges to $u_i^*$.} This allows \green{the nodes} to
determine if the average is exactly \red{$\bar x$} (integer average),
or if it lies in \green{$(\bar u_i^*,\bar u^*_i+1)$, or $(\bar u^*_i-1,\bar
u^*_i)$, which is the property asserted by}  Theorem
\ref{thm:average_computation}.

\begin{small}

\end{small} 


\end{document}